\pdfoutput=1

\documentclass[sigconf]{aamas} 

\usepackage{balance} 
\usepackage{amsthm}
\newtheorem{definition}{\textbf{Definition}}
\newtheorem{prop}{\textbf{Proposition}}
\newtheorem{theorem}{\textbf{Theorem}}
\newtheorem{lemma}{\textbf{Lemma}}
\newtheorem{corollary}{\textbf{Corollary}}
\usepackage{enumerate}

\usepackage{algorithm}
\usepackage{algorithmicx}
\usepackage{algpseudocode}
\usepackage{amsmath}
\usepackage{makecell}
\usepackage{multirow}
\usepackage{subfigure}

\title[AAMAS-2023 Formatting Instructions]{A Generic Multi-Player Transformation Algorithm for Solving Large-Scale Zero-Sum Extensive-Form Adversarial Team Games}

\author{Chen Qiu}
\affiliation{
  \institution{Harbin Institute of Technology, ShenZhen}
  \city{Shenzhen}
  \country{China}}
\email{19b951015@stu.hit.edu.cn}

\author{Yulin Wu}
\affiliation{
  \institution{Harbin Institute of Technology, ShenZhen}
  \city{Shenzhen}
  \country{China}}
\email{yulinwu@cs.hitsz.edu.cn}

\author{Weixin Huang}
\affiliation{
  \institution{Harbin Institute of Technology, ShenZhen}
  \city{Shenzhen}
  \country{China}}
\email{21s051015@stu.hit.edu.cn}

\author{Botao Liu}
\affiliation{
  \institution{Harbin Institute of Technology, ShenZhen}
  \city{Shenzhen}
  \country{United Kingdom}}
\email{21s151108@stu.hit.edu.cn}

\author{Shaohuai Shi}
\affiliation{
  \institution{Harbin Institute of Technology, ShenZhen}
  \city{Shenzhen}
  \country{China}}
\email{shaohuais@hit.edu.cn}

\author{Xuan Wang}
\affiliation{
  \institution{Harbin Institute of Technology, ShenZhen}
  \city{Shenzhen}
  \country{China}}
\email{wangxuan@cs.hitsz.edu.cn}

\begin{abstract}

Many recent practical and theoretical breakthroughs focus on adversarial team multi-player games (ATMGs) in \textit{ex ante} correlation scenarios. In this setting, team members are allowed to coordinate their strategies only before the game starts. Although there existing algorithms for solving extensive-form ATMGs, the size of the game tree generated by the previous algorithms grows exponentially with the number of players. Therefore, how to deal with large-scale zero-sum extensive-form ATMGs problems close to the real world is still a significant challenge. In this paper, we propose a generic multi-player transformation algorithm, which can transform any multi-player game tree satisfying the definition of AMTGs into a 2-player game tree, such that finding a team-maxmin equilibrium with correlation (TMECor) in large-scale ATMGs can be transformed into solving NE in 2-player games. To achieve this goal, we first introduce a new structure named \textit{private information pre-branch}, which consists of a temporary chance node and coordinator nodes and aims to make decisions for all potential private information on behalf of the team members. We also show theoretically that NE in the transformed 2-player game is equivalent TMECor in the original multi-player game. This work significantly reduces the growth of action space and nodes from exponential to constant level. This enables our work to outperform all the previous state-of-the-art algorithms in finding a TMECor, with $182.89\times$, $168.47\times$, $694.44\times$, and $233.98\times$ significant improvements in the different Kuhn Poker and Leduc Poker cases (21K3, 21K4, 21K6 and 21L33). In addition, this work first practically solves the ATMGs in a 5-player case which cannot be conducted by existing algorithms.

\end{abstract}


\newcommand{\BibTeX}{\rm B\kern-.05em{\sc i\kern-.025em b}\kern-.08em\TeX}

\begin{document}


\pagestyle{fancy}
\fancyhead{}


\maketitle 


\section{Introduction}
Games have been critical testbeds for exploring how effectively machines can make sophisticated decisions since the early days of computing \cite{bard2020hanabi}. Finding a equilibrium in games has become a significant criterion for evaluating artificial intelligence levels. In recent years, many great results have been obtained from the field of 2-player zero-sum (2p0s) games based on Nash Equilibrium (NE) \cite{nash1951non} in non-complete information environments \cite{zinkevich2007regret, DBLP:conf/icml/BrownLGS19, DBLP:conf/iclr/ZhouRLYZ20}. However, solving equilibrium in multi-player zero-sum games with three or more players remains a tricky challenge. There are three main reasons for this: Firstly,  the CFR-like algorithms for finding NE are widely used in 2p0s games, but no theoretical guarantees are provided in the literature whether they can be directly used in multi-player games \cite{brown2020equilibrium}; Secondly, NEs are not unique in multi-player games, and the independent strategies of each player cannot easily form a unique NE \cite{brown2019superhuman}; Thirdly, computing NEs is PPAD-complete for multi-player zero-sum games \cite{chen20053}.

The Team-maxmin Equilibrium (TME) \cite{von1997team, DBLP:conf/aaai/BasilicoCN017} is a solution concept that can handle multi-player games. In this paper, we are concerned with adversarial team multi-player games. It models a situation in which a team of players shares the same utility function against an adversary. Based on the various forms of correlation between team members, this concept is extended to extensive-games \cite{DBLP:conf/aaai/Celli018}. Notably, we focus on the \textit{ex ante} coordination scenario of team members. More specifically, team members are allowed to coordinate and agree on a common strategy before the game starts, but they cannot communicate during the game. The variant of TME in this scenario is called Team-maxmin Equilibrium with Correlation (TMECor), and its computation is shown to be FNP-hard \cite{hansen2008approximability}. To the best of our knowledge, Celli and Gatti \cite{DBLP:conf/aaai/Celli018} first proposed a linear programming algorithm capable of solving TMECor in 2018. The essence of this algorithm is hybrid column generation, where the team members and adversary use different forms of strategy representation. The Associated Recursive Asynchronous Multiparametric Disaggregation Technique (ARAMDT) proposed by \cite{DBLP:conf/aaai/Zhang020} also uses mixed-integer linear program (MILP) to find the TMECor. Since the feasible solutions of MILP are too large for large-scale games, this significantly limits their ability to efficiently handle large-scale games.

The TMECor can be considered as an NE between the team and the adversary that maximizes the utility of the team. In adversarial team multi-player games, the advantage of TMECor over NE is unique, considerably reducing the difficulty of finding optimal strategies. Nonetheless, there are fewer algorithms for TMECor. Therefore, it is an interesting and worthwhile research topic to utilize the algorithms (e.g., CFR, CFR+, etc.) of solving Nash equilibrium in two-player zero-sum extensive-form games for finding a TMECor. Inspired by \cite{DBLP:conf/icml/CarminatiCC022}, we attempt to make a connection between large-scale adversarial team multi-player games and 2-player zero-sum games, achieving better performance and faster convergence than the state-of-the-art algorithms.

\textbf{Main Contributions.} The most outstanding outcome of our work is a generic multi-player transformation algorithm (MPTA) that can convert a tree of adversarial team multi-player games (ATMGs) into a tree of 2-player zero-sum games with theoretical guarantees. Thus, the classical and efficient algorithms in 2-player zero-sum games can be used to solve TMECor in large-scale ATMGs (e.g., CFR \cite{zinkevich2007regret}, CFR+ \cite{DBLP:conf/ijcai/TammelinBJB15}, DCFR \cite{DBLP:conf/aaai/BrownS19}, MCCFR \cite{lanctot2009monte}, etc.). One of the reasons why large-scale ATMGs are difficult to solve is that the action space grows exponentially with the increase in players. To design this algorithm, we propose a new structure, calling it \textit{private information pre-branch}, which can reduce the growth of action space from exponential to a constant level. Theoretically, the more players, the more obvious the effect of reducing the action space. Therefore, the transformed game tree of our method is smaller in size compared to similar algorithms, such that it can speed up the computation of TMECor and make it possible to use it in larger scale team games, which are closer to real-world problems. Furthermore, to provide a primary theoretical guarantee for our work, we prove the equilibrium equivalence of the game before and after the transformation. Finally, the performance of our method has been proven to be excellent through multiple sets of experiments. The experimental results show that the team-maxmin strategy profile obtained by our algorithm is closer to TMECor than the baseline algorithm at the same time and significantly reduces the running time in the same iteration rounds.


\section{Related Work}
Many studies have focused on adversarial team multi-player games (ATMGs) in an attempt to find a solution since the concept of Team-maxmin Equilibrium (TME) was introduced in 1997. Kannan et al. \cite{kannan2002strategic} adapt the idea of a team game and develop an algorithm for finding optimal paths based on information networks. Hansen et al. \cite{DBLP:conf/wine/HansenHMS08} prove that the task of obtaining a TME is FNP-hard and calculate the theoretical time complexity for the first time. According to the communication capabilities of the team members, Celli and Gatti \cite{DBLP:conf/aaai/Celli018} define three different scenarios and corresponding equilibriums for the first time in the extensive-form ATMGs. In particular, it is the first time that computing the TMECor of ATMGs (i.e., the equilibrium in the scenario where team members are allowed to communicate before the game starts) using a column generation algorithm combined with hybrid representation. 

Thereafter, \cite{DBLP:conf/aaai/Zhang020, DBLP:conf/icml/0004020, DBLP:conf/aaai/00010C21} propose a series of improvements to the Hybrid Column Generation (HCG) algorithm. The main disadvantage of this method is that the feasible solution space of integer or mixed-integer linear program is too large, which severely limits the size and speed of games it can solve. Basilico et al. \cite{DBLP:journals/ia/BasilicoCNG17} propose a modified version of the quasi-polynomial time algorithm and an algorithm named \textit{IteratedLP} which is a novel anytime approximation algorithm. \textit{IteratedLP}'s working principle is to maintain the current solution, which provides a policy that can be returned at any time for each team member. Farina et al. \cite{farina2018ex} adopt a new \textit{realization form} representation for mapping the problem of finding an optimal ex-ante-coordinated policy for the team to the problem of finding NE. Zhang and Sandholm \cite{DBLP:conf/aaai/ZhangS22} devise a tree decomposition algorithm for solving team games. To reduce the number of constraints required, the authors use a tree decomposition for constraints and represent the team's strategy space by the correlated strategies of a polytope. Since the team need to sample a strategy policy for each player from a joint probability distribution, Farina et al. \cite{DBLP:conf/icml/FarinaC0S21} propose a modeling on computing the optimal distribution outcome and allowing the team to obtain the highest payoff by increasing the upper limit on the number of strategy files. Cacciamani et al. \cite{DBLP:conf/atal/CacciamaniCC021} and celli et al. \cite{celli2019coordination} use multi-agent reinforcement learning approaches. The former adds a game-theoretic centralized training regimen and serves as a buffer of past experiences. Unfortunately, these reinforcement learning methods can only be applied in particular circumstances. The idea of a team being represented by a single coordinator as used by Carminati et al. \cite{DBLP:conf/icml/CarminatiCC022} is closely related to ours. However, the size of game tree generated by the algorithm used by the authors grows exponentially as the action space increases. This situation makes it difficult to use this algorithm for large-scale team games.

In our work, we propose an innovative approach for transforming a team multi-player game tree into a 2-player game tree, where the 2-player game tree is constructed in a form that is distinct from the method by \cite{DBLP:conf/icml/CarminatiCC022}. In the converted game tree, the coordinator represents the strategy in the same way as the team members. At the same time, the number of new nodes is greatly reduced, allowing the game tree to avoid exponential increases in size.


\section{Preliminaries}

This section briefly introduces some of the basic concepts and definitions used in this paper. To learn more details, see also \cite{DBLP:conf/aaai/Celli018, zinkevich2007regret}. For clarity and intuition, detailed descriptions of some variables are shown in Table ~\ref{table1}.

\begin{table}[t]
  \caption{Detailed descriptions of some variables}
  \label{table1}
  \begin{tabular}{ll}\toprule
    \textit{Variable} & \textit{Detailed description}  \\ \midrule
    $I_{p}(h)$ & \makecell[l]{The information set of player $p$ at decision node \\ $h \in H$.} \\
    $\Pi_{p}(z)$ & \makecell[l]{The set of the player $p$'s normal-form plans where $p$ \\ can reach the terminal node $z \in Z$.} \\
    $\sigma_{p}(h)$ & \makecell[l]{The strategy of player $p$ at decision node $h \in H$.} \\
    $\mu_{p}\left[\pi_{p}\right]$ & \makecell[l]{The probability that player $p \in P$ will follow the \\ actions specified by the normal-form plan $\pi_{p}$.}  \\
    $u_{\mathcal{T}}(z)$ & \makecell[l]{The shared utility of the team with reaching \\ terminal node $z \in Z$.}   \\
    $\Delta\left(\Pi_p\right)$ & \makecell[l]{The formulaic definition of normal-form strategies,\\ i.e.,  the probability distribution of the normal-form \\ plans  of player $p \in P$.}  \\ \bottomrule
  \end{tabular}
\end{table}


\subsection{Extensive-Form Games and Nash Equilibrium}

An extensive-form game $G$ is the tree-form model of imperfect-information games with sequential interactions \cite{kuhn1950extensive, brown2017safe}.

\begin{figure*}
    \centering
    \includegraphics[width=1\textwidth]{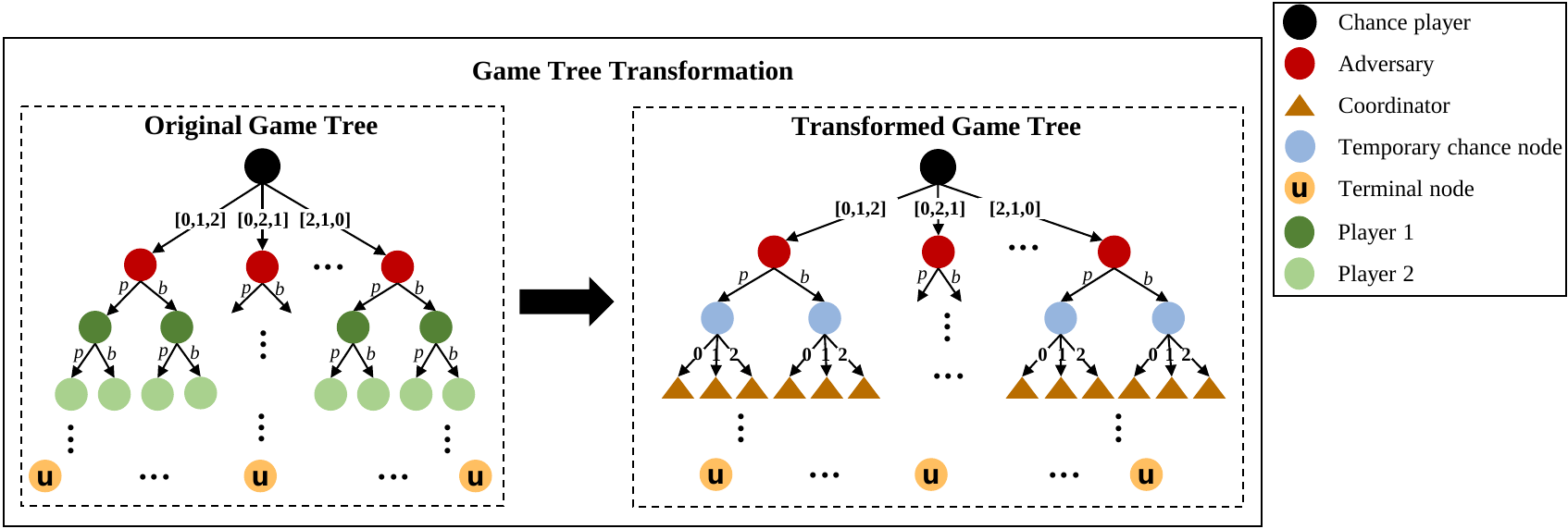}
    \caption{An example of the 3-player game tree transformation process. The dashed box on the left shows a 3-player game tree which fits the definition of ATG and includes player 1, player 2 and the adversary. Player 1 and player 2 form a team against a common opponent. The 2-player game tree after MPTA transformation is displayed in the dashed box on the right. $[0, 1, 2]$ denotes cards assigned to players. $p$ and $b$ are the available actions of players (i.e., "pass" and "bet"). }
    \label{transProcess}
\end{figure*}

\begin{definition}[Extensive-form Games] 
A finite extensive-form game $G$ is a tuple $\left\langle N, A, H, Z, \mathcal{A}, P, u, I\right\rangle$:
\begin{enumerate}
    \item A set $N$ of players, $N=1, 2,..., n$;
    \item A set $A$ of all the possible actions, $A=\cup_{p \in N} A_p$, where $A_p$ denotes a set of available actions of player $p$;
    \item A set $H$ is the set of nodes in the game tree, $h \in H$. $h$ can also be represented action histories (sequence of actions from root node to current node $h$); 
    \item A set $Z$ contains all the leaf nodes of game tree, $Z \subseteq H$;
    \item For each decision node $h \in H$, the result returned by function $\mathcal{A}(h)$ is all available actions at node $h$;
    \item Given a node $h$, the function $P(h)$ is player who takes an action after node $h$;
    \item For each player $p \in N$, the utility function $u_p$  is the payoff mapped from the terminal node $z \in Z$ to the reals $\mathbb{R}$;
    \item A set $I$ belonging to player $p\in N$, all nodes $h, h^{\prime} \in I$ are indistinguishable to $p$ in $I_p$.
\end{enumerate}
\end{definition}

The Nash equilibrium is a common solution concept in zero-sum extensive-form games and is a strategy profile for player $i$ denoted as $\sigma^*_i$. Given the strategy profile for $n$ players, $\sigma_1, \sigma_2, \ldots, \sigma_n$, an NE can be formally represented as
\begin{eqnarray}
u_i(\sigma) & \geq & \max _{\sigma_i^* \in \Sigma_i} u_i\left(\sigma_i^*, \sigma_{-i}\right)
\end{eqnarray}

\subsection{Adversarial Team Multi-Player Games and Team-Maxmin Equilibrium with Correlation}

In this paper, we concentrate on extensive-form adversarial team multi-player games. Formally, an adversarial team multi-player game (ATMG) $\mathcal{G}$ has $N (N \geq 3)$ players, in which a team consisting of $N-1$ team members against independently an adversary $\mathcal{O}$. We refer to $\mathcal{C}$ and $\mathcal{T}$ as the chance player and team (coordinator) respectively. We set up the scenario restricted to a zero-sum extensive-form ATMG, where $u_{\mathcal{T}} = -u_{\mathcal{O}}$. Regardless of whether the team wins or loses, team members share benefits equally with the same utility function in ATMG. That is, $u_i(l)=u_j(l), \forall i, j \in T$.  Moreover, this work will focus on games with perfect recall, where all players are able to recall their previous actions and the corresponding information sets.

\begin{definition}[Behavioral Strategy] 
A behavioral strategy $\sigma_{p}$ of player $p\in N$ is a function that assigns a distribution over all the available actions $\mathcal{A}\left(I_p\right)$ to each $I_p$.
\end{definition}

 A behavioral strategy profile $\sigma$ is composed of each player's behavioral strategy, where $\sigma = \{\sigma_{1}, \sigma_{2}, \cdots, \sigma_{n}\}$. The extensive-form ATMGs also provide additional strategy representation:

\begin{definition}[Normal-Form Plan] 
A normal-form plan $\pi_{p} = \times_{I_{p}} \mathcal{A}(I_{p})$ of player $p$ is a tuple specifying one action for each information set of player $p$.
\end{definition}

Furthermore, the normal-form strategy of player $p \in N$ is denoted as $\mu_p$, which is the probability distribution over the normal-form plans. Similarly, a normal-form strategy profile is $\mu$ $=$ $\{\mu_{1},$ $\mu_{2},$ $\cdots,$ $\mu_{n}\}$. Given a normal-form strategy $\mu_{p}$ of player $p \in N$, $\mu_{-p}$ refers to all normal-form strategies in $\mu$ except $\mu_{p}$. A TMECor is proven to be a Nash equilibrium (NE) which maximizes the team's utility \cite{DBLP:conf/ijcai/00010S22, von1997team}. Concerning ATMG settings, TMECor differs from NE in that it always exists and is unique. The team members use behavioral strategies during the TME calculating procedure. This is due to the lack of necessity for coordination among team members. However, our work focus on the scenario of \textit{ex ante} correlation. If the behavioral strategy is still adopted, the correlation between the normal-form strategies of team members cannot be accurately obtained because of a lack of coordination \cite{DBLP:conf/nips/FarinaC0S18}. Therefore, in order to compute TMECor, it is necessary for team members to adopt normal-form strategies.

A TMECor is able to find through a linear programming formulated over the normal-form strategy profile of all players:
\begin{equation} \label{equation1}
\begin{split}
\max _{\mu_{\mathcal{T}}} \min _{\mu_{\mathcal{O}}} \sum_{z \in Z} & \sum_{\substack{p \in \mathcal{T}  \\ \pi_{p} \in \Pi_{p}(z) \\ \pi_{\mathcal{O}} \in \Pi_{\mathcal{O}}(z)}} \mu_{\mathcal{T}}\left[\pi_{p}\right] \mu_{\mathcal{O}}\left[\pi_{\mathcal{O}}\right] u_{\mathcal{T}}(z) \\
\text { s.t. } & \mu_{\mathcal{T}} \in \Delta\left(\times_{p \in \mathcal{T}}{ } \Pi_p\right) \\
& \mu_{O} \in \Delta\left(\Pi_{O}\right)
\end{split}
\end{equation}


\section{Method}

\subsection{Overview}

To solve the adversarial team multi-player games, the core of our work is to transform an ATMGs-based multi-player game into a 2-player game, then use the well-established algorithms in 2-player games to find a NE equivalent to the TMECor, as shown in Figure \ref{Framework}. For achieving this purpose, we first construct a new structure representation for the team member nodes in the original multi-player game tree. Secondly, we utilize this structure to design an algorithm named MPTA, which is used to transform a original multi-player game tree into a 2-player game tree. The MPTA consists of two phases: 1) traversing over the whole original game tree and transforming all nodes to obtain a 2-player game tree; 2) merging information sets with temporary private information and pruning for 2-player game tree.

\begin{figure}[hpb]
    \centering
    \includegraphics[width=0.45\textwidth]{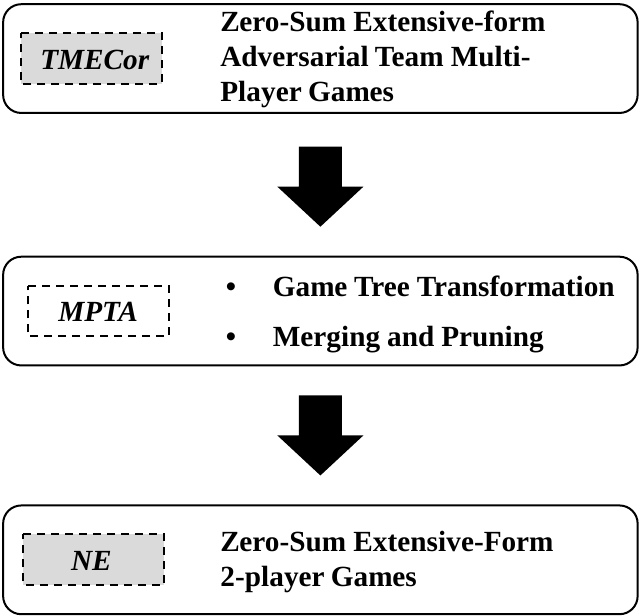}
    \caption{An overview of our method.}
    \label{Framework}
\end{figure}

\subsection{The Structure of Private Information Pre-Branch}
We add coordinator players and temporary opportunity nodes to the MPTA process, and they form a new structure, which we call the \textit{private information pre-branch} (PIPB), as shown in the virtual triangle box in Figure \ref{transformed}. The coordinator player only makes decisions for one team member at a time, so we give him all possible information of hand cards in advance. The extensive-form game has property with sequence decision, and we assume that the adversary takes action first in the game process. The parent of each coordinator player node must be a temporary chance node according to this structure. In other words, the coordinator makes decisions based on the structure of PIPB. 

\begin{prop}[growth of action space]
    Given a transformed game $\mathcal{G}$, each occurrence of a coordinator player node $c$ in the game tree increases the action space by $\lvert cards \rvert \times \lvert A \rvert$, where $\lvert cards \rvert$ indicate the number of cards.
\end{prop}

\begin{algorithm}[t]
\caption{Multi-Player Game Tree Transformation Algorithm}
\label{algorithm1}
\begin{algorithmic}[1] 
\Function {MPGTT}{$G$}
\State  \Comment{$G=  \left\langle N, A, H, Z, \mathcal{A}, P, u, I\right\rangle$ is original game tree}
\State initialize $\mathcal{G}$ new ATMG tree
\State $Rank \gets Rank(G)$
\State $origRoot \gets Root(G)$
\State $newRoot \gets Root(\mathcal{G})$
\State $\mathcal{G} \gets ProcOfTrans(origRoot, newRoot, \mathcal{G}, Rank)$
\State \Return{$\mathcal{G}$}
\EndFunction 

\Function {ProcOfTrans}{$origNode, parNode, newTree, rank$}
\If {$origNode \in Z$}  \Comment{for leaf node}
\State addLeaf($origNode$)
\ElsIf {$P(origNode)=\{ \mathcal{C} \}$} \Comment{for chance node}
\For {$childNodes$ in $origNode$}
\State ProcOfTrans($childNodes$, $newRoot$, $newTree$,  \\ $rank$)
\EndFor
\ElsIf {$P(origNode=\{ \mathcal{O} \})$}   \Comment{for adversary node }
\State $newNode \gets origNode$
\State $addnode(newNode)$
\For {$childNodes$ in $origNode$}
\State ProcOfTrans($childNodes$, $newNode$, $newTree$, $rank$)
\EndFor

\Else  \Comment{for team member node}
\State $addTempChanceNode(origNode)$
\State $newNode \gets origNode$
\For {$r \in rank$}
\State $tempChanceNode \gets$ a new node whose parent is newNode
\For {$childNodes$ in $origNode$}
\State ProcOfTrans($childNodes$, $tempChanceNode$, $newTree$, $rank$) 
\EndFor
\EndFor
\EndIf
\State \Return{$newTree$}
\EndFunction
\end{algorithmic}
\end{algorithm}


\subsection{Phase 1: Game Tree Transformation}

The primary distinction between the MPTA and previous approaches is that the coordinator player in the transformed game tree represents a team member who is playing rather than the entire team. Team members have both private information, e.g., the hand cards, and public information that can be observed by other players, including adversary, e.g., the game history. Note that the coordinator player only knows the private information of a team member currently represented and public information of the current situation, but not the other players' private information. The pseudo-code of the transformation process is depicted in Algorithm \ref{algorithm1}.

\begin{figure}[hpb]
    \centering
    \includegraphics[width=0.45\textwidth]{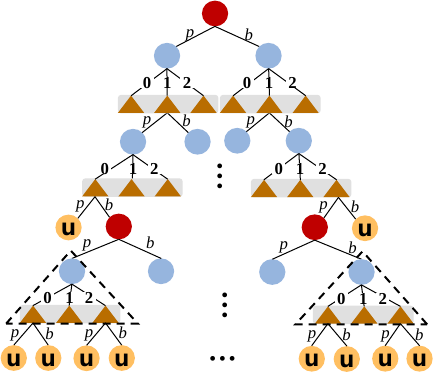}
    \caption{An example of transformed game tree's two branches omitting the chance node, with 10 action sequences: $[ppp]$, $[bbb]$, $[ppbpp]$, $[ppbpb]$, $[ppbbp]$, $[ppbbb]$, $[bbppp]$, $bbppb$, $[bbpbp]$, and $[bbpbb]$. The triangular dashed box shows one of the private information pre-branching structures.}
    \label{transformed}
\end{figure}

Since our algorithm is built on a game tree, we first construct a complete game tree based on an improved generic game scenario (e.g., Kuhn Poker \cite{kuhn1950simplified}), as illustrated in the dashed box on the left of Figure \ref{transProcess}. The chance player who deals the cards is the root node in the original game tree, and its branches represent different dealing situations. As the number of cards rises, the number of chance player's child nodes often increases as well. The nodes below the root node are the players' decision nodes, where the players choose their actions in turn according to the principles of a sequential game. The payoff for a round will be calculated once all players have finished their actions and reached the terminal nodes. In particular, team members share the benefits whether they win or lose because they use the same utility function in ATMG.

Then, an original game tree $G$ will be passed as a parameter to the MPTA. We add each of its branches to the transformed game tree $\mathcal{G}$ for the root node. Player decision nodes in $G$ can be divided into two categories: decision nodes of team members and an adversary. The adversary decision nodes are unmodified. That is, the same branches are built in $\mathcal{G}$ for adversary nodes. We introduce a coordinator player to replace the team member who is taking action. More specifically, we transform the team member decision nodes into the coordinator player, who provides the current player's strategies. However, the coordinator player is not aware of the team member's hand cards when first choosing an action on behalf of the team member. Thus, we construct temporary chance nodes to represent all potential situations for a team member's hand cards, as shown in Figure \ref{transformed}. We design the \textit{information pool} for storing information of team members, which will contribute to optimizing our approach in Subsection \ref{4.4}.

\subsection{Phase 2: Merging and Pruning} \label{4.4}

\paragraph{\textbf{Merging information sets with temporary private information.}} The information set, which includes both the public action histories and the private information of players, is an indispensable property of imperfect information games. Each branch of the root node represents a different distribution of cards. In the original game tree, there are no multiple nodes belonging to the same information set within the same branch. On the contrary, the condition is quite dissimilar after transformation. We begin at the root node and select two branches of the transformed game tree where one player has the same hand cards, except for temporary chance nodes. The nodes of players are traversed in turn. The player's nodes on the same layer belong to a single information set in the case of same game histories, even though they have different virtual private information. In our work, nodes are given the same number to mark the same information set.

\paragraph{\textbf{Pruning for transformed tree.}} Under the same branch of root node, the coordinator does not know private information if and only if he makes the first decisions for each team member in turn. Their private information is stored in the \textit{information pool} when all team members have been replaced by the coordinator. Then, the coordinator will be allowed to extract the current team member's hand cards from \textit{information pool} to make the next decisions. During the transformation process, the coordinator is only authorized to extract the private information of one player at a time, which means that the private information of team members is independent of each other. In this way, we can prune the transformed game tree to make its scale smaller, as shown in Algorithm \ref{algorithm2}.


\begin{algorithm}[t]
\caption{Pruning Algorithm}
\label{algorithm2}
\begin{algorithmic}[1] 
\Function {Prun}{$rootNode$}
\State  \Comment{$rootNode \in \mathcal{G}$}
\For {$childNode$ in $root$}
\If {$childNode \in Z$}
\State \Return{}
\ElsIf {$childNode \in \mathcal{C}$ or $\mathcal{O}$ or $T$}
\State Prun($childNode$)
\ElsIf {$childNode \in \mathcal{C}_t$}
\If {$Sequence(\mathcal{T}) \neq \emptyset$}
\For {$a$ in $\mathcal{A}(childNode)$}
\If {$a \neq cards[P(childNode(a))]$}
\State remove($childNode(a)$)
\EndIf
\EndFor
\Else
\State $Pruned(childNode)$
\EndIf
\EndIf
\EndFor

\EndFunction
\end{algorithmic}
\end{algorithm}


\section{Equilibrium Equivalence Proof}
As a primary guarantee for this work, we prove theoretically that game trees before and after the transformation are essentially equivalent. It is further reasoned that TMECor in the adversarial team multi-player games is equivalently related to NE in 2-player zero-sum extensive-form games. For simplicity, we denote the original game tree by $G$ and the transformed game tree $\mathcal{G}$ is represented as the return value of $MPTA(G)$ in Algorithm \ref{algorithm1}. $\sigma^{*}$ is used to represent the TMECor in $G$. Given a player $p$, $\sigma_{p}$ denotes the strategy in $G$ and $\mu_{p}$ denotes the normal-form strategy in $\mathcal{G}$. Let $H_{\sigma_{p}}$ be the set of decision nodes reached according to player $p$'s strategy $\sigma_{p}$ and $c$ is the coordinator player in $\mathcal{G}$. Note that temporary chance node $\mathcal{C}_{t}$ is not the players' decision node. The equilibrium equivalence will be proved according to the following lemmas and theorems.

\begin{lemma}\label{lemma1}
    Given a multi-player game tree $G$ that satisfies the definition of ATMGs, and the transformed game tree $\mathcal{G} = MPTA(G)$, for any player $p$, any decision node $h_{p}$ in $G$ can be mapped to the decision node $h^{\prime}_{p}$ in $\mathcal{G}$. Formally, $\forall p \in N$ in $G$, $h \in H$ and $h^{\prime} \in H^{\prime}$, $H_{p} \subsetneqq H^{\prime}_{p}$.
\end{lemma}

\begin{proof}
    We can prove that Lemma \ref{lemma1} holds according to the process of transformation in Algorithm \ref{algorithm1}, using the characteristics of the game tree structure. To avoid notational confusion, let $Z$ and $Z^{\prime}$ represent the set of leaf nodes in $G$ and $\mathcal{G}$, respectively. 
    \begin{itemize}
        \item \textbf{For leaf node:} because the leaf nodes correspond to the players' payoff, they are added directly to the game tree without any changes.  Formally, $Z = Z^{\prime}$.
        
        \item \textbf{For adversary node:} $\sigma_{\mathcal{O}}$ in $G$ is mapped to $\mu_{\mathcal{O}}$ in $\mathcal{G}$ by our method, which provides a guarantee for the equality of $H_{\sigma_{\mathcal{O}}}$ and $H_{\mu_{\mathcal{O}}}$.
        
        \item \textbf{For chance node:} the root node of the game tree is not changed during the transformation. So $\mathcal{C}$ in $G$ and $\mathcal{C}^{\prime}$ in $\mathcal{G}$ are same.
        
        \item \textbf{For team member node:} For any decision node of team member $p$, there is always a corresponding decision node of coordinator player $c$ with the same game histories in the transformed game tree.
        
    \end{itemize}
    Since the coordinator player does not know the private information of all players, temporary chance node of the PIPB structure in $\mathcal{G}$ will provide the coordinator player with all potential situations of hand cards.

    This is the end of the proof.
\end{proof}

\begin{lemma} \label{lemma2}
    Given a multi-player game tree $G$ that satisfies the definition of ATMGs, and the transformed game tree $\mathcal{G} = MPTA(G)$, for any player $p$, any decision node $h_{p}$ in $\mathcal{G}$ can maintain the correspondence with the decision node $h^{\prime}_{p}$ in $G$. Formally, $\forall p \in N$ in $G$, $h \in H$ and $h^{\prime} \in H^{\prime}$, $H_{p} \subsetneqq H^{\prime}_{p}$.
\end{lemma}

\begin{proof}
    The proof process is similar to that of Lemma \ref{lemma1}, and the above process is not repeated in this proof. Furthermore, the coordinator nodes added by $\mathcal{C}_{t}$ in $\mathcal{G}$ belongs to the same information set. Therefore, their correspondence with decision nodes in $G$ still holds.

    This is the end of the proof.
\end{proof}

\begin{corollary} \label{corollary1}
    Given a multi-player game tree $G$ that satisfies the definition of ATMGs, and the transformed game tree $\mathcal{G} = MPTA(G)$, for any player $p$ in the original game, his strategies must be mapped in the transformed game. Formally, 

    $$
    \begin{aligned}
        &\forall p \in \mathcal{T}: \sigma_{p}(h) = \mu_{c}(h) \\
        &  \sigma_{\mathcal{O}}(h) = \mu_{\mathcal{O'}}(h) \\
        & \sigma_{\mathcal{C}} = \mu_{\mathcal{C}}
    \end{aligned}
    $$
    
\end{corollary}

\begin{proof}
    We can learn from Lemma \ref{lemma1} and Lemma \ref{lemma2} that the strategies of players in the original game and the transformed game are equivalent.
\end{proof}

\begin{theorem} \label{theorem1}
    Given a multi-player game tree $G$ that satisfies the definition of ATMGs, and the transformed game tree $\mathcal{G} = MPTA(G)$, for any team member $p$ and opponent player $\mathcal{O}$, we have:
    $$
    \begin{aligned}
    &\forall p \in N_{G}, c \in N_{\mathcal{G}} : u_{p}(\sigma)= u_{c}(\mu)  \\
    & \mathcal{O} \in N_{G}, \mathcal{O}^{\prime} \in N_{\mathcal{G}} : u_{\mathcal{O}}(\sigma) = u_{\mathcal{O}^{\prime}}(\mu)
    \end{aligned}
    $$
where the strategy $\mu$ is a mapping of strategy $\sigma$ from $G$ to $\mathcal{G}$.
\end{theorem}

\begin{proof}
    The proof can be obtained from Lemma \ref{lemma1}, Lemma \ref{lemma2} and Corollary \ref{corollary1} that the payoffs of the players are equivalent in the original game and transformed game. In fact, the terminal nodes representing the utility is added directly to $\mathcal{G}$ according to the transformation steps in Algorithm \ref{algorithm1}.
\end{proof}

\begin{theorem} \label{theorem2}
    Given a multi-player game tree $G$ that satisfies the definition of ATMGs, and the transformed game tree $\mathcal{G} = MPTA(G)$, a Nash equilibrium in $\mathcal{G}$ has equilibrium equivalence with TMECor in $G$. Formally, $\mu^{*}_{c} = \sigma^{*}_{\mathcal{T}}$.
\end{theorem}

\begin{proof}
    Assuming that $\sigma^{*}_{\mathcal{T}}$ is a TMECor, then according to Equation \ref{equation1}, we can get:

    $$
    \begin{aligned}
      \sigma^{*}_{\mathcal{T}} \in \arg \max _{\sigma_{\mathcal{T}}} \min _{\sigma_{\mathcal{O}}} \sum_{z \in Z} & \sum_{\substack{p \in \mathcal{T}  \\ \pi_{p} \in \Pi_{p}(z) \\ \pi_{\mathcal{O}} \in \Pi_{\mathcal{O}}(z)}} \sigma_{\mathcal{T}}\left[\pi_{p}\right] \sigma_{\mathcal{O}}\left[\pi_{\mathcal{O}}\right] u_{\mathcal{T}}(z)
    \end{aligned}
    $$
    According to Corollary \ref{corollary1}, the above formula can be converted into 
    $$
    \begin{aligned}
      \mu^{*}_{c} \in \arg \max _{\mu_{c}} \min _{\mu_{\mathcal{O}}} \sum_{z \in Z} & \sum_{\substack{ \pi_{c} \in \Pi_{c}(z) \\ \pi_{\mathcal{O}} \in \Pi_{\mathcal{O}}(z)}} \mu_{c}\left[\pi_{c}\right] \mu_{\mathcal{O}}\left[\pi_{\mathcal{O}}\right] u_{c}(z)
    \end{aligned}
    $$
    $\min _{TME \operatorname{Cor}}\left(\sigma_{\mathcal{T}}\right)$ and $\min _{NE}\left(\mu_c\right)$ denote the inner minimization problem in the definition of TMECor and NE respectively. We assume that existing a $\mu_{c}^{\prime}$ that is larger than the value of $\sigma_{\mathcal{T}}^{*}$, i.e., $\min_{NE}(\mu_{c}^{\prime}) > \min_{NE}(\sigma^{*}_{\mathcal{T}})$. By Theorem \ref{theorem1}, it can be converted that: $\min_{TME}(\mu^{\prime}_{c}) > \min_{TME}(\sigma^{*}_{\mathcal{T}})$.
    This is unreasonable since by hypothesis $\sigma_{\mathcal{T}^{*}}$ is a maximum.
    Hence, we get that:
    $$
    \begin{aligned}
      \mu_{c}^{*} \in \arg \max_{p \in \mathcal{T}} \min_{TME}(\sigma_{T}\left[ \pi_{p} \right])
    \end{aligned}
    $$
    This is the end of the proof.

\end{proof}


\section{experimental evaluation}
In this section, we describe the setup of experimental scenarios and comparison methods. The performance of our method is verified by comparing with the state-of-the-art algorithm in different scenarios.

\subsection{Experimental Setting}

We conduct our experiments with standard testbed in the adversarial team multi-player games (ATMGs). More exactly, we use modified versions of \textit{Kuhn Poker} \cite{kuhn1950simplified} and \textit{Leduc Poker} \cite{DBLP:conf/uai/SoutheyBLPBBR05}. By taking the number of players and cards as parameters and modifying the team's utility function, they are possible to satisfy the definition of ATMGs. We set up scenarios in which a group of players forms a team against a single player and unify the team's utility function to meet the definition of ATMGs. Furthermore, we can achieve the purpose of generating multiple experimental platforms with different complexity by taking the number of players, suits, and cards as parameters. We adopt a total of 12 scenarios of varying difficulty, 6 each for \textit{Kuhn Poker} and \textit{Leduc Poker}, where the default maximum number of bets allowed per betting round is 1. In \textit{Kuhn Poker}, the experiment is set up with players from 3 to 5, and cards are set to 3, 4, 6, and 8. In \textit{Leduc Poker}, players are set from 3 to 5, cards are set by 3, 4, and 6, suits are set 3. It is worth noting that the more complex 5-player scenario has not been attempted before. All experiments are run on a machine with 18-core 2.7GHz CPU and 250GB memory.

We use the \textit{Counterfactual Regret Minimization plus} (CFR+) to interface with the MPTA and the baseline algorithm, which is an excellent approach for finding Nash equilibrium in 2-player zero-sum extensive-form games. Nevertheless, our method can still use other CFR-like algorithms. In this paper, the baseline algorithm is Team-Public-Information Conversion Algorithm (TPICA), which is the previous best method similar to our work. It gives the coordinator information that is common to the whole team and provides all team members with the corresponding actions for each possible private state.

\subsection{Experimental Results}

\paragraph{\textbf{Total runtime for finding a TMECor.}}
In Table \ref{res_table}, we describe in detail the size of the game trees in different scenarios, both the original game tree and the transformed tree through our method. Furthermore, they are compared with the basic method (TPICA) proposed by Carminati et al. \cite{DBLP:conf/icml/CarminatiCC022}. TPICA and our work share the same goal of finding a TMECor in ATMGs using the effective tools of 2p0s games. The data in Table \ref{res_table} shows that our method significantly reduces the size of the game trees generation after the transformation compared to TPICA. In \textit{Kuhn Poker}, the sizes of 21K3, 21K4, and 21K6 are reduced by $9.25 \times$, $431.72 \times$, and $1,476.26 \times$, respectively. The blank cells in Table \ref{res_table} indicate that TPICA fail to convert due to out-of-memory and thus cannot obtain a valid game tree. In the four cases of 21K4, 21K4, 21K6, and 21L33 where both MPTA and TPICA can work, the total time required by our approach to compute a TMECor is 0.76s, 9.26s, 144s, and 240s respectively, which is $182.89 \times$, $168.47 \times$, $694.44 \times$, and $233.98 \times$ faster than TPICA. This shows that our method is effective in reducing the action space, as it improves the solving speed by several orders of magnitude. It is worth noting that MPTA is still available in other large-scale scenarios where TPICA cannot transform original game trees. In particular, 41K6 and 41L33 are 5-player cases that have never been used as experiments by previous algorithms due to their sheer size. In addition, we also observe that the reason for the speed-up is mainly due to the special structure, which greatly reduces the number of adversary nodes and temporary chance nodes.

\begin{figure}[t]
\centering
\subfigure[21K3]{
\label{res1_1}
\includegraphics[width=0.231\textwidth]{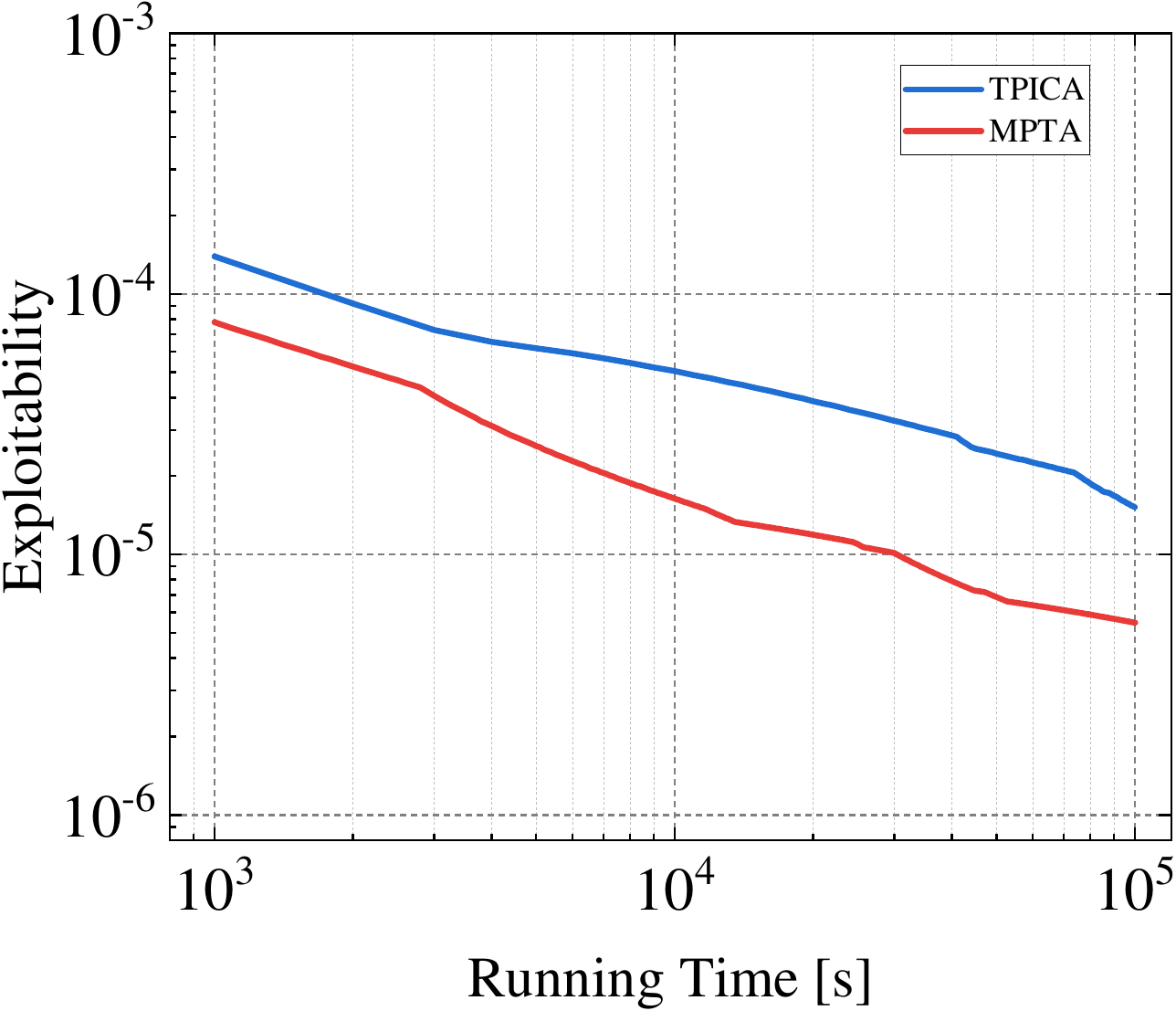}}
\subfigure[21L33]{
\label{res1_2}
\includegraphics[width=0.231\textwidth]{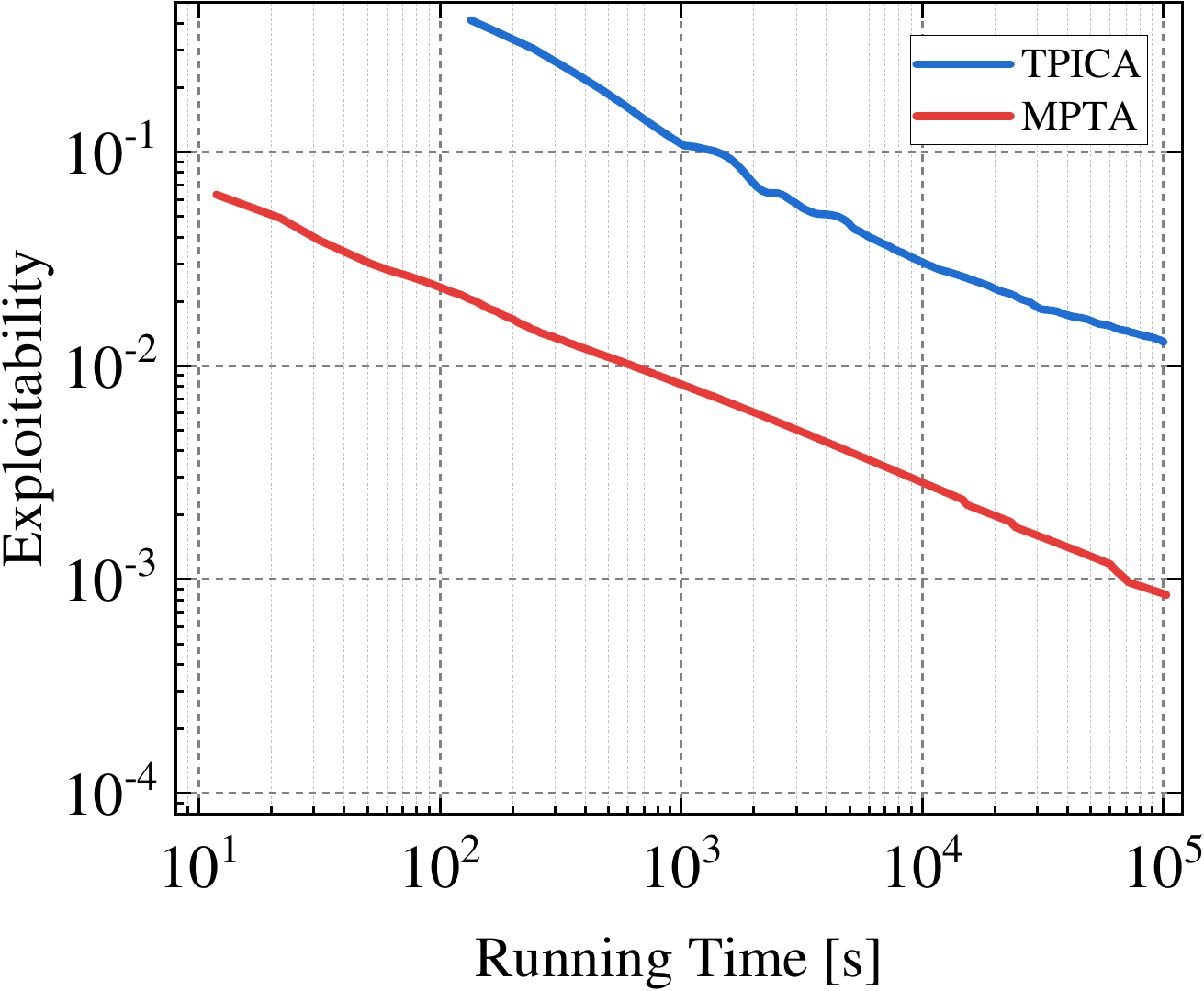}}
\subfigure[21K6]{
\label{res1_3}
\includegraphics[width=0.231\textwidth]{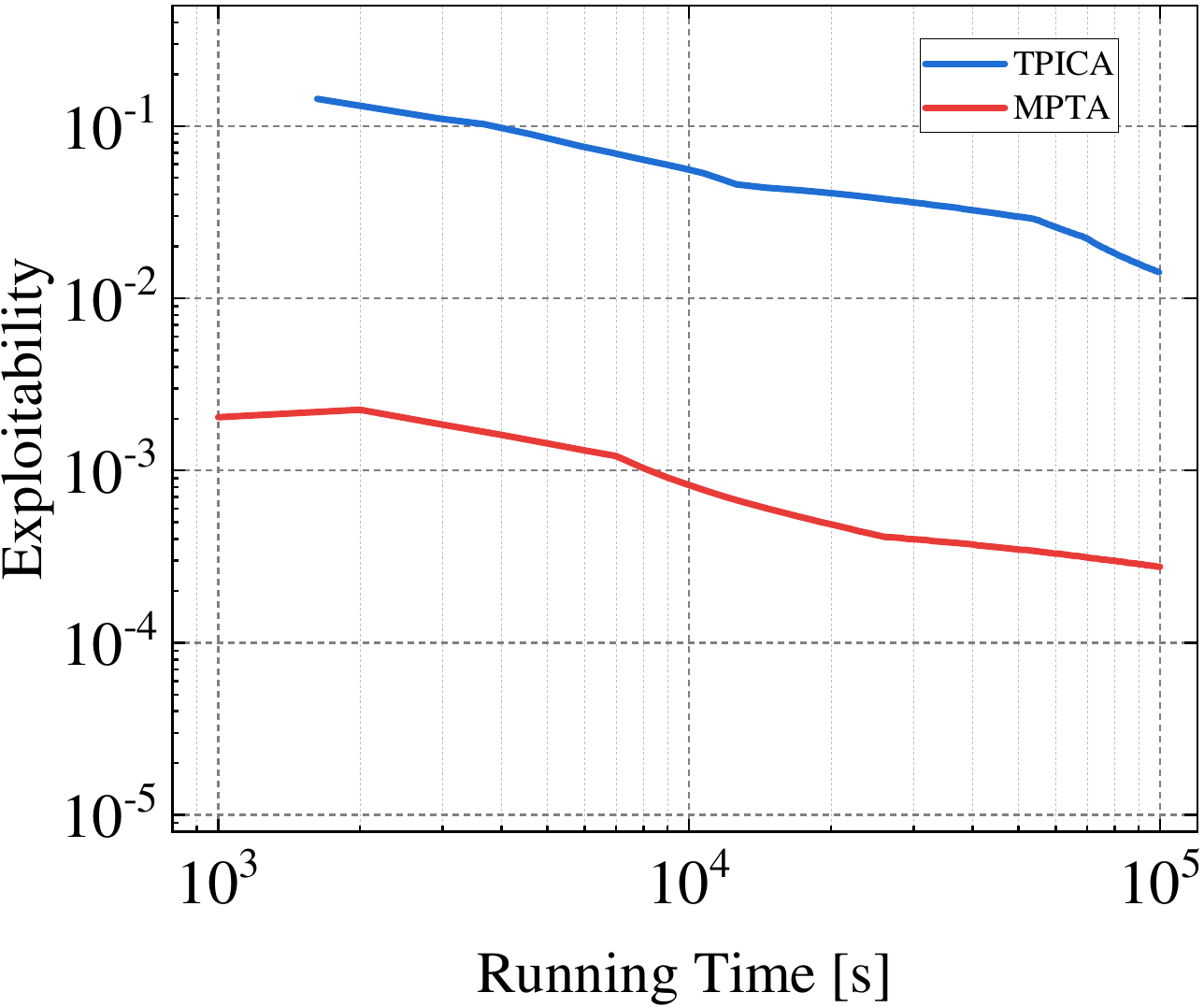}}
\subfigure[31L43]{
\label{res1_4}
\includegraphics[width=0.231\textwidth]{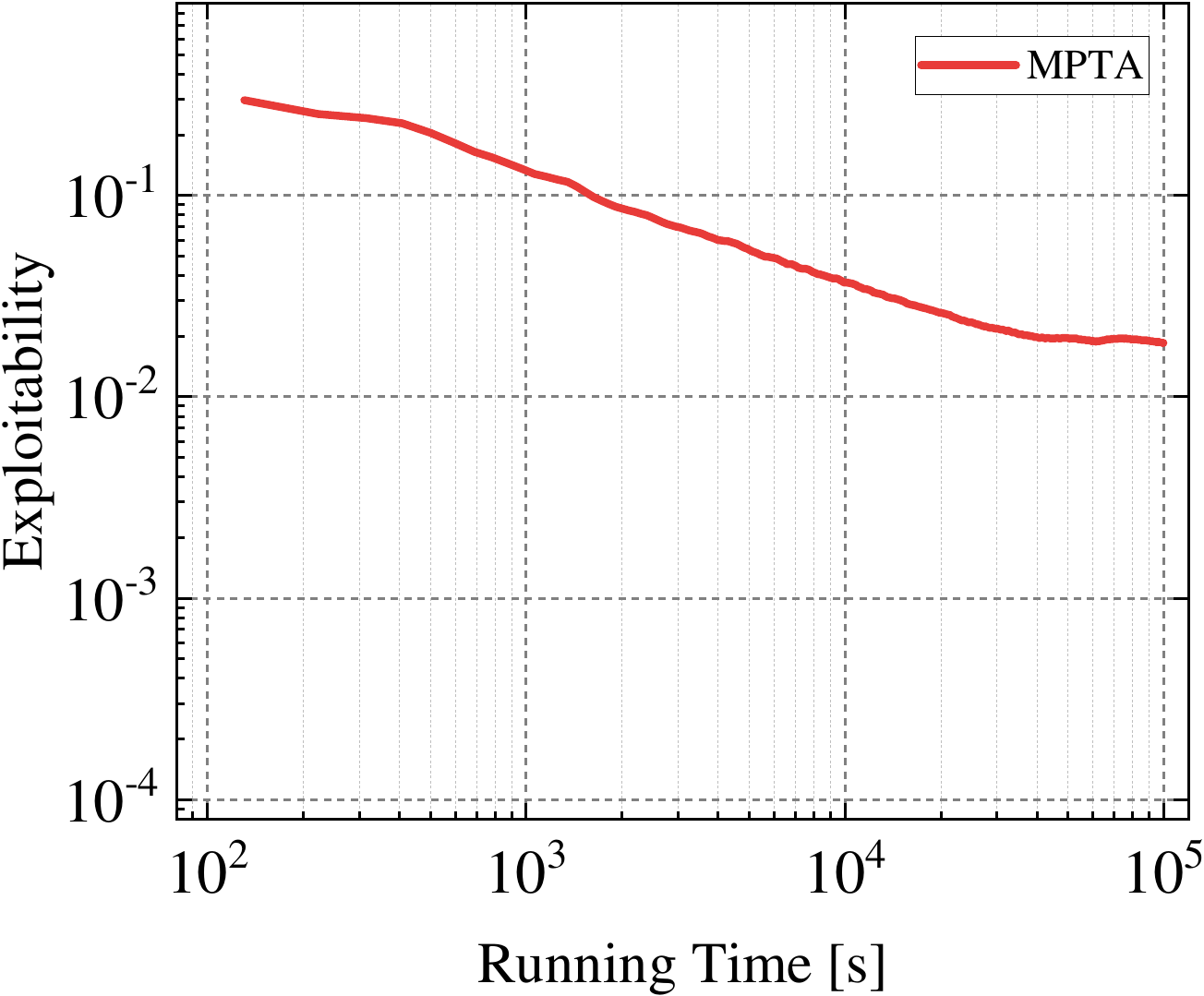}}
\subfigure[41K6]{
\label{res1_5}
\includegraphics[width=0.231\textwidth]{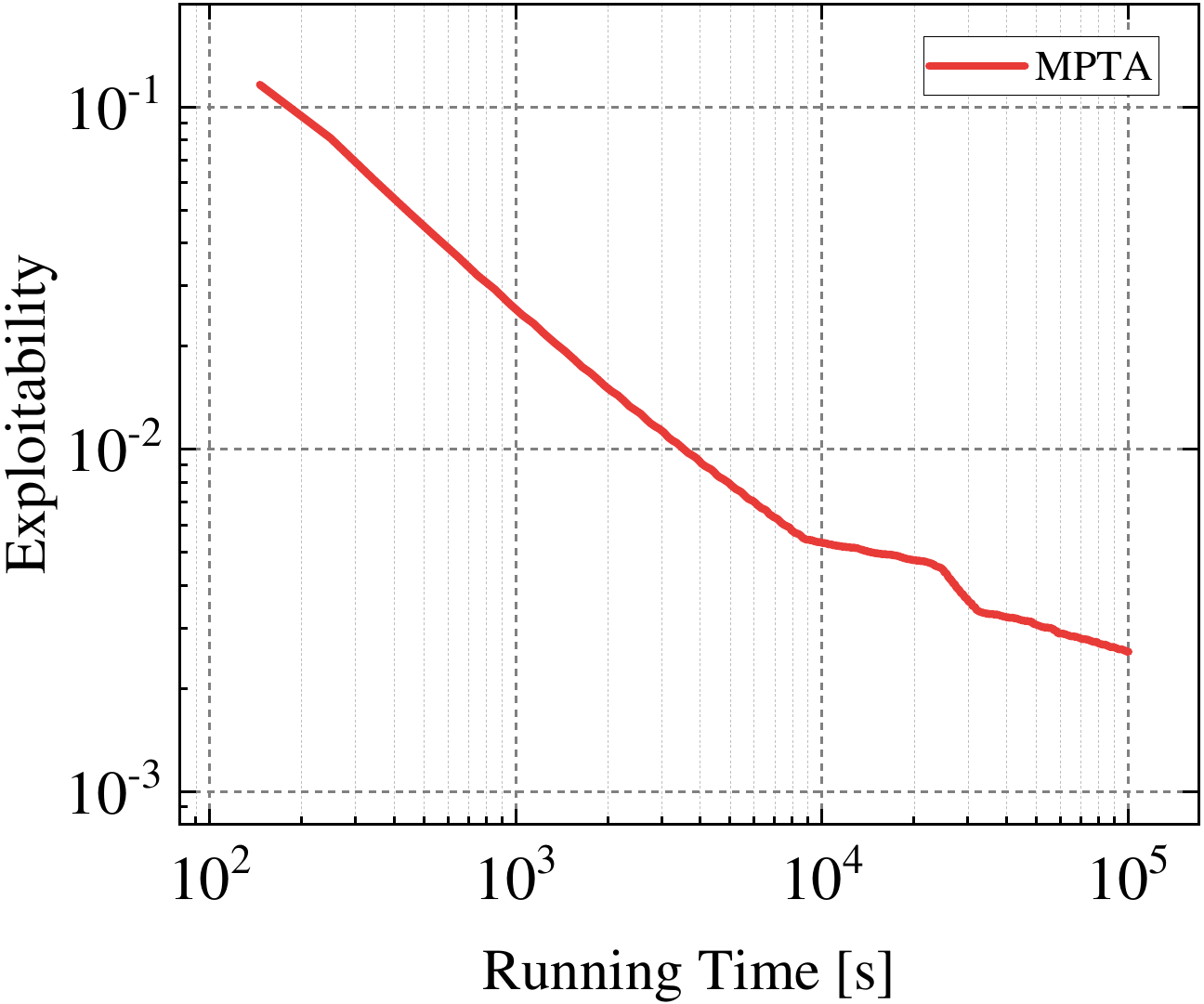}}
\subfigure[41L33]{
\label{res1_6}
\includegraphics[width=0.231\textwidth]{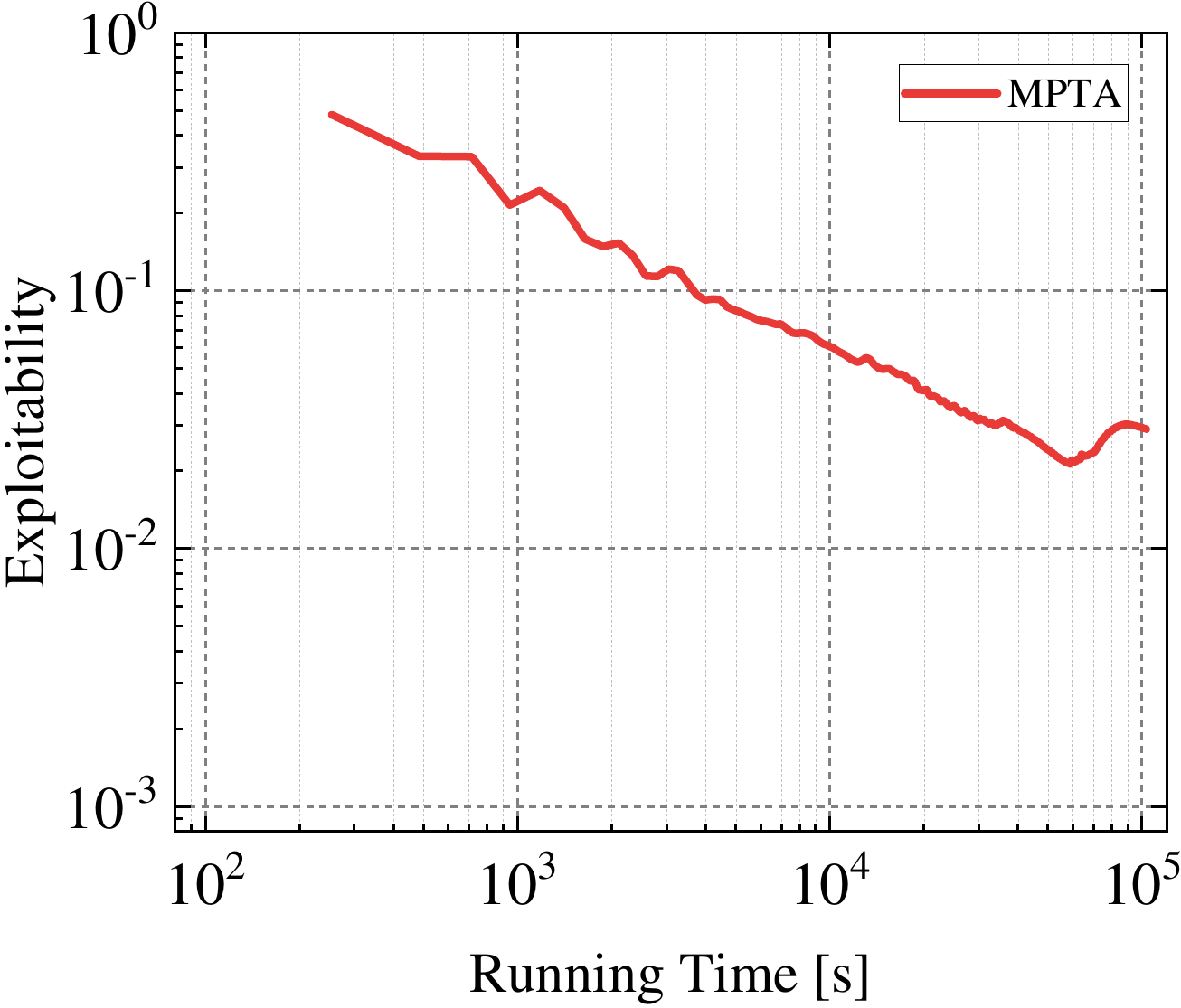}}

\caption{The comparison of exploitability in the same running time. All experiments run for 100,000 seconds. In 41K6, 31L43, and 41L33, TPICA failed to transform due to out-of-memory.}  
\label{res1}
\end{figure}

\begin{table*}[htbp]
  \centering
  \caption{The size of the game trees returned by the original game, transformed game by our method and by the basic method by Carminati et al. \cite{DBLP:conf/icml/CarminatiCC022}. under the different complexity scenarios; the runtime of finding a TMECor by our method and the baseline algorithm in different games. The last column shows exactly the improvement of our method compared to the baseline.\textbf{(The last column shows exactly how much faster our method is compared to baseline.)} $mn$K$r$ is \textit{Kuhn Poker}, where $m$ is the number of team members, $n$ is the number of adversary and $r$ is the number of cards. $mn$L$rc$ is \textit{Leduc Poker}, where $m$ is the number of team members, $n$ is the number of adversary, $r$ is the number of cards and $c$ is the indistinguishable suits.}
    \begin{tabular}{lrrrrrrrrrr}
    \toprule
    \multicolumn{1}{c}{\multirow{2}[2]{*}{Game instances}} & \multicolumn{3}{c}{Total nodes} & \multicolumn{2}{c}{Team nodes} & \multicolumn{2}{c}{Adversary nodes} & \multicolumn{2}{c}{Runtime} & \multicolumn{1}{c}{\multirow{2}[2]{*}{Improvements}} \\
          & \multicolumn{1}{c}{Original} & \multicolumn{1}{c}{TPICA} & \multicolumn{1}{c|}{\textbf{MPTA}} & \multicolumn{1}{c}{TPICA} & \multicolumn{1}{c|}{\textbf{MPTA}} & \multicolumn{1}{c}{TPICA} & \multicolumn{1}{c|}{\textbf{MPTA}} & \multicolumn{1}{c}{TPICA} & \multicolumn{1}{c}{\textbf{MPTA}} &  \\
    \midrule
    21K3  & 151   & 5,395 & 583   & 300   & 144   & 294   & 72    & 139s  & \textbf{0.76s} & \textbf{182.89$\times$} \\
    21K4  & 601   & 1,337,051 & 3,097 & 3,888 & 768   & 4,632 & 384   & 1560s & \textbf{9.26s} & \textbf{168.47$\times$} \\
    21K6  & 3,001 & 34,191,721 & 23,161 & 261,360 & 5,760 & 368,760 & 2,880 & >27h  & \textbf{144s} & \textbf{694.44$\times$}  \\
    31K6  & 23,401 &       & 271,441 &       & 75,240 &       & 22,680 &       & \textbf{825s} &  \\
    31K8  & 109,201 &       & 1,713,601 &       & 475,440 &       & 142,800 &       & \textbf{5,093s} &  \\
    41K6  & 115,921 &       & 1,796,401 &       & 528,480 &       & 105,120 &       & \textbf{3,051s} &  \\
    21L33 & 13,183 & 10,777,963 & 57,799 & 614,172 & 14,664 & 475,566 & 6,864 & 56,156s & \textbf{240s} & \textbf{233.98$\times$}  \\
    21L43 & 42,589 &       & 251,749 &       & 64,008 &       & 29,736 &       & \textbf{3,006s} &  \\
    21L63 & 218,011 &       & 1,954,351 &       & 497,940 &       & 229,620 &       & \textbf{9,024s} &  \\
    31L33 & 161,491 &       & 948,151 &       & 262,500 &       & 80,220 &       & \textbf{4,014s} &  \\
    31L43 & 738,241 &       & 5,994,241 &       & 1,661,760 &       & 504,000 &       & \textbf{137,817s} &  \\
    41L33 & 1,673,311 &       & 12,226,231 &       & 3,535,320 &       & 809,880 &       & \textbf{143,475s} &  \\
    \bottomrule
    \end{tabular}%
  \label{res_table}%
\end{table*}%

\paragraph{\textbf{Solving efficiency.}}
Exploitability is widely used as a significant evaluation criterion for a strategy profile. Informally, it represents the gap between the current policy and the optimal policy. Thus, a smaller exploitability indicates that the current strategy is closer to the TMECor in the original multi-player games. In addition to the comparison of the total solution time, we also need insight into the relationship between the solving efficiency of the transformed game trees and algorithms in the running process. For this purpose, we select three cases each in \textit{Kuhn Poker} and \textit{Leduc Poker} to test MPTA and baseline method's variation of exploitability over time within a limited running time of 100,000 seconds, as shown in Figure \ref{res1}. In the three comparable cases: 21K3, 21K6, and 21L33, we observe that the curve representing MPTA is always below the TPICA, and the distance between the two curves increases with the increase of the size gap of game trees. This suggests that the CFR+ is more efficient in solving the transformed game tree for MPTA. In the remaining scenarios, we can find that the convergence rate to equilibrium is faster in 41K6 than in 31L43 and 41L33 and the curve representing the MPTA fluctuates more sharply in 41L33. This indicates that the process of computing equilibrium is more difficult when the game scale increases.

\balance

\paragraph{\textbf{Execution efficiency.}}
Finding a TMECor is an iterative process, and we show a comparison of the time taken by the algorithms within the same number of iteration rounds in Figure \ref{res2}. As can be seen from Figure \ref{res2}, in all comparing circumstances, the game trees transformed by MPTA take considerably less time to complete the computation of a normal-form strategy profile than by TPICA in any number of iteration rounds. 21K6 and 21L33 are an order of magnitude larger than 21K3 and 21K4. However, the results in Figure \ref{res2_3} and Figure \ref{res2_4} show that the advantages of our approach are more evident in these two large-scale games.

\begin{figure}[htbp]
\centering
\subfigure[21K3]{
\label{res2_1}
\includegraphics[width=0.231\textwidth]{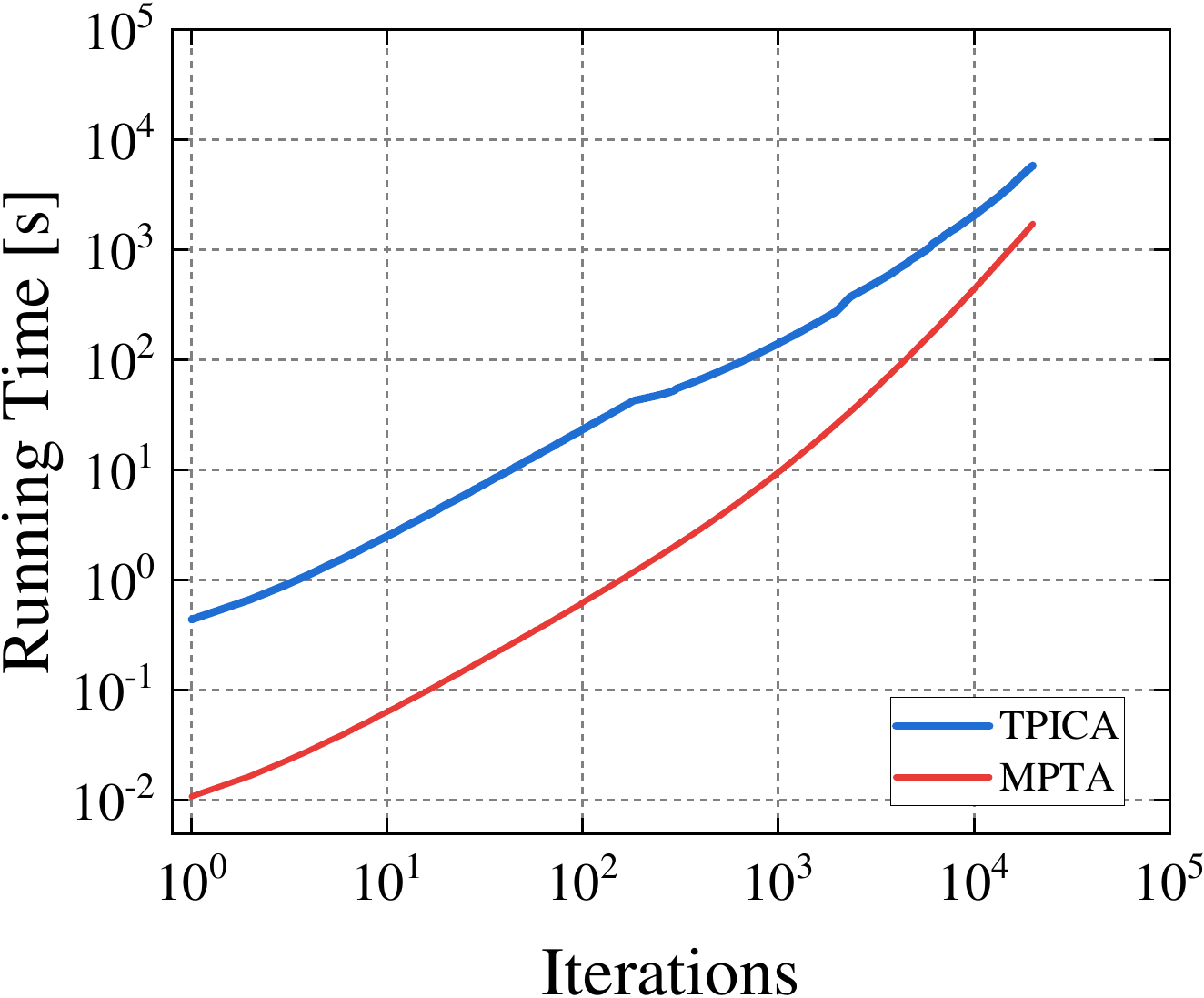}}
\subfigure[21K4]{
\label{res2_2}
\includegraphics[width=0.231\textwidth]{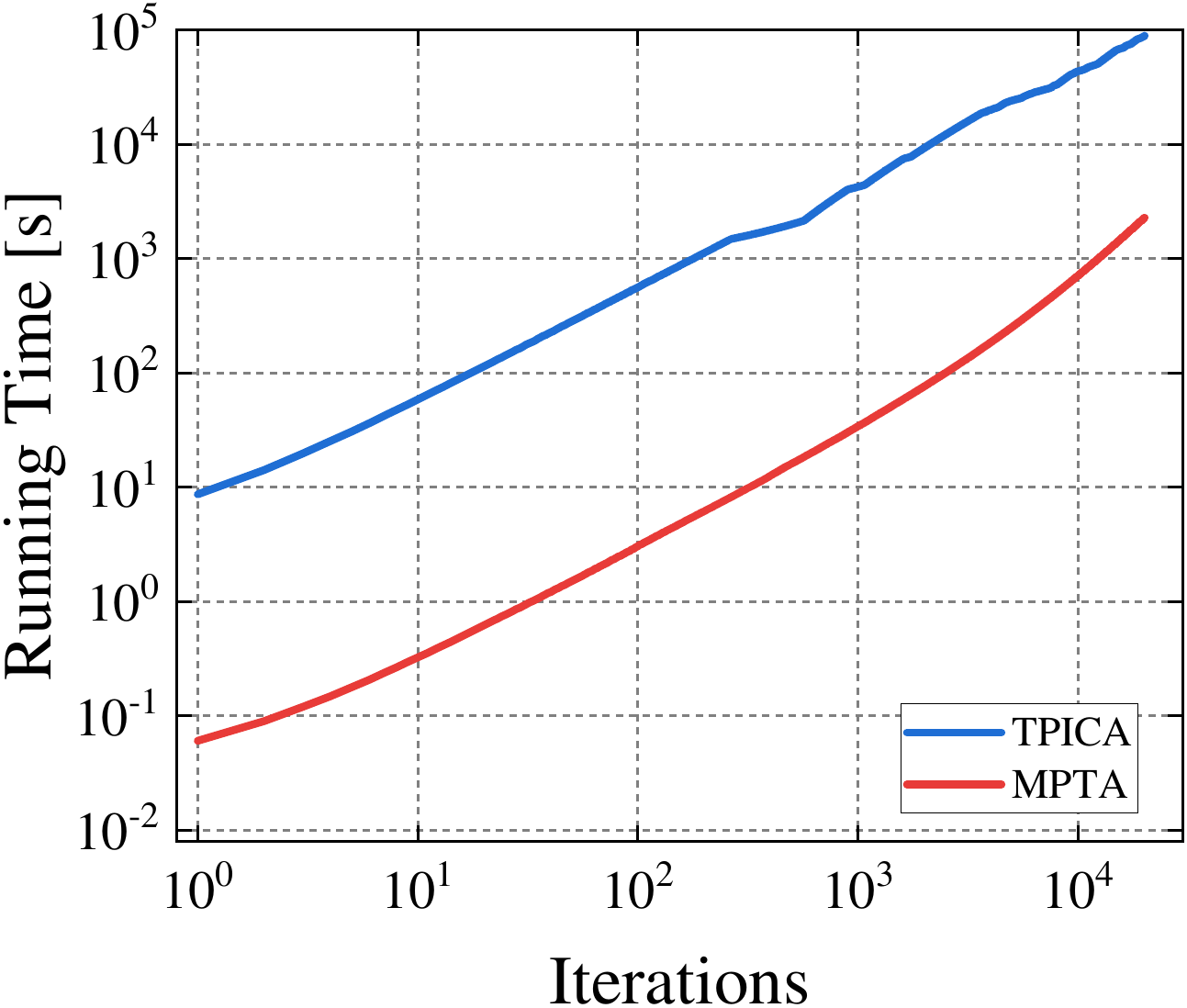}}
\subfigure[21K6]{
\label{res2_3}
\includegraphics[width=0.231\textwidth]{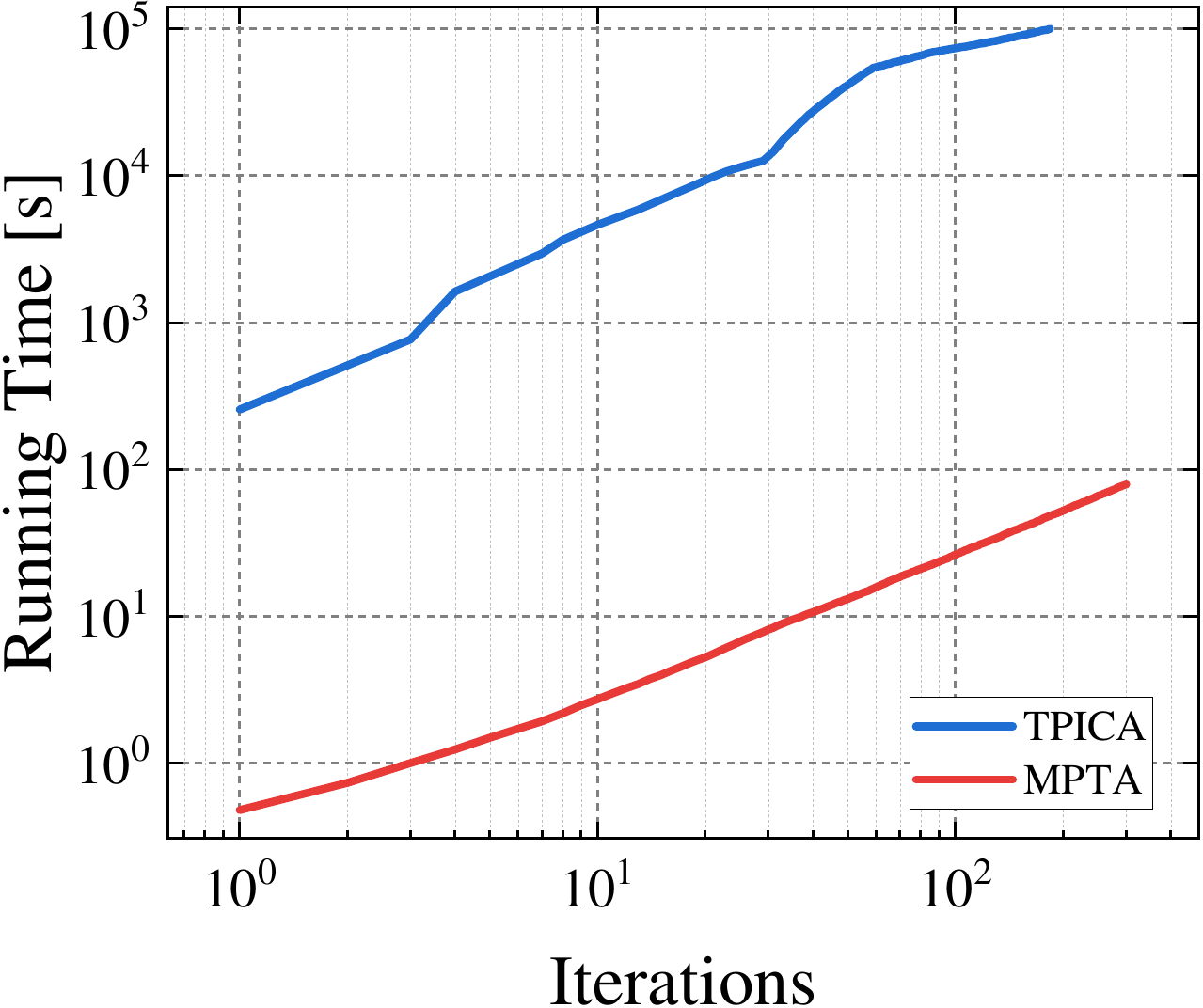}}
\subfigure[21L33]{
\label{res2_4}
\includegraphics[width=0.231\textwidth]{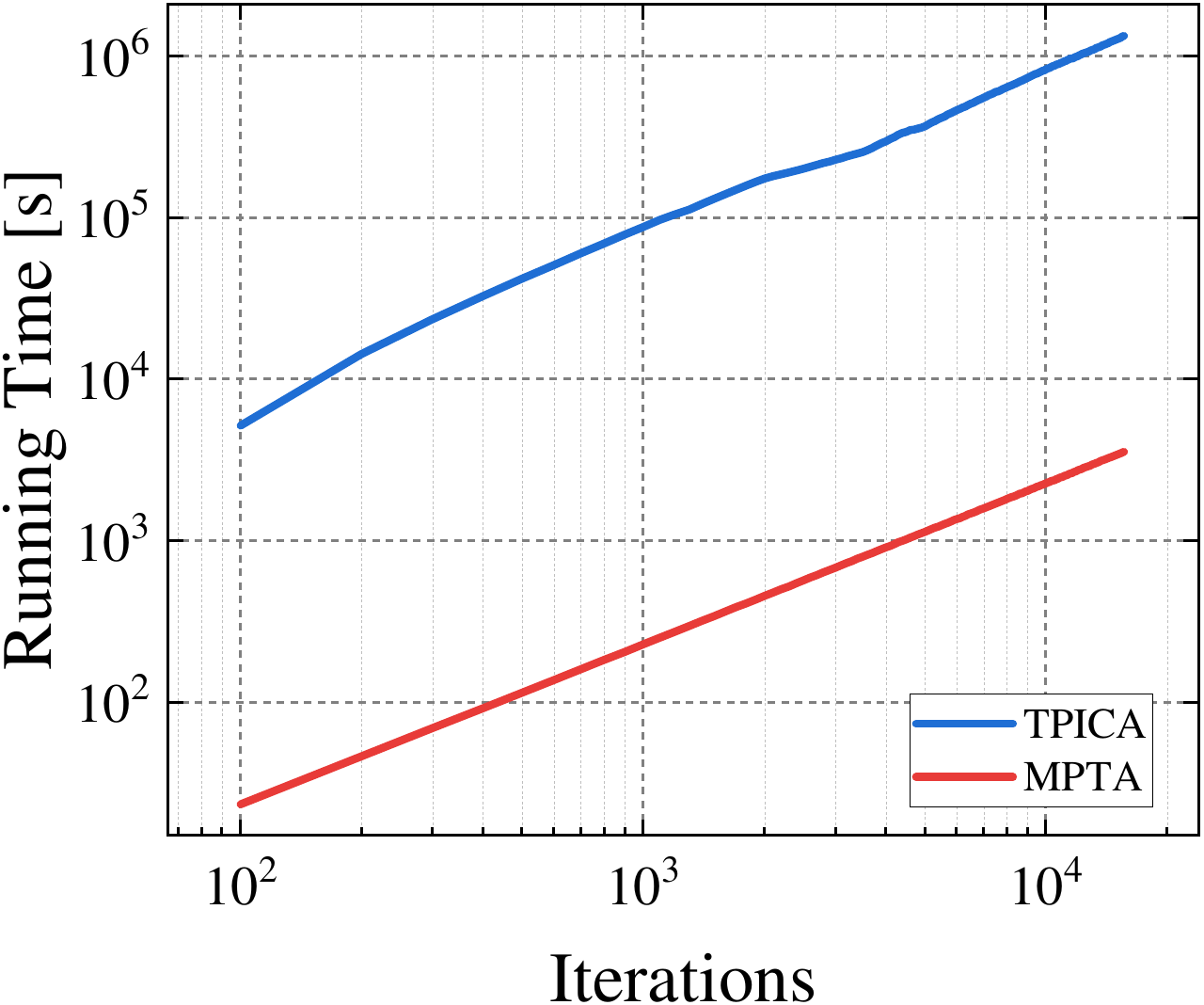}}
\caption{The comparison of the time taken by the algorithms within the same number of iterations. All schemes except for 21K6 have been iterated for 20,000 rounds, as the TPICA is too time-consuming to run more rounds.}  
\label{res2}
\end{figure}


\section{conclusions and future work}
In this paper, we present a generic multi-player transformation algorithm (MPTA) which can transform a multi-player game tree satisfying the definition of ATMGs into a 2-player game tree, thereby establishing a bridge between 2p0s games and multi-player games. In addition, we analyze theoretically that the proposed new structure limits the growth of the transformed game's action space from exponential to a constant level. At the same time, we also prove the equilibrium equivalence between the original game tree and the transformed game tree, which provides a theoretical guarantee for our work. We experiment with several scenarios of varying complexity and show that finding a TMECor is several orders of magnitude faster than the state-of-the-art baseline. As far as we know, this work is the first one to solve a ATMG with 5 or even more players.

In the future, we will be devoted to testing our algorithm in more challenging and complex scenarios in real world. We also plan to  get the aid of deep neural networks' powerful data processing and generalization capabilities to  provide real-time strategies for agents.






\bibliographystyle{ACM-Reference-Format} 
\bibliography{mybibfile}


\end{document}